\documentclass[a4paper,11pt]{paper}

\usepackage[english]{babel}
\usepackage[T1]{fontenc}
\usepackage[utf8x]{inputenc}
\usepackage{microtype}
\usepackage{graphicx}
\graphicspath{{./figs/}}
\usepackage{calc}

\usepackage{amssymb,amsthm,amsmath}
\newtheorem{theorem} {Theorem}[section]
\newtheorem{lemma}[theorem]{Lemma}

\usepackage{thmtools,thm-restate}
\usepackage{xspace}

\newcommand{\Caratheodory}{Carath\'eodory\xspace}
\newcommand{\Mat}{Matou\v sek}

\newtheorem*{azuma}{Azuma's Inequality}

\theoremstyle{plain}
\newtheorem{proposition}[theorem]{Proposition}
\newtheorem{ob}[theorem]{Observation}

\newcommand{\R}{\mathbb{R}}
\newcommand{\N}{\mathbb{N}}
\newcommand{\eps}{\varepsilon}
\newcommand{\mc}[1]{\mathcal{#1}\xspace}
\DeclareMathOperator{\conv}{conv}
\newcommand{\dual}[1]{\overline{#1}}

\newcommand{\tO}{\tilde{O}}
\newcommand{\ti}{^{(i)}}
\newcommand{\tii}{^{(i+1)}}
\newcommand{\lar}{{\mbox{\tiny large}}}
\newcommand{\E}{\mathbb{E}}
\newcommand{\IGNORE}[1]{}
\newcommand{\ignore}[1]{}
\newcommand{\ori}{\mathbf{0}}
\newcommand{\ray}[1]{\overrightarrow{#1}}

\newcommand{\hdtime}{\tilde{O}(n^{d/2 + 1})}

\title{Approximating the Simplicial Depth}

\author{%
  Peyman Afshani \textsuperscript{1}\thanks{Work supported in part by the Danish National Research
  Foundation grant DNRF84 through Center for Massive Data Algorithmics
(MADALGO).}\and
Donald R.\ Sheehy \textsuperscript{2}
\and
Yannik Stein \textsuperscript{3}\thanks{Supported by the Deutsche
    Forschungsgemeinschaft within the research training group ``Methods for
    Discrete Structures'' (GRK 1408).}
}

\newcommand{\instsep}[1]{\makebox[2em]{#1}}
\newcommand{\einstsep}[1]{\makebox[2em]{}}
\institution{%
  \instsep{1}MADALGO, Department of Computer Science,\\ 
  \einstsep{}Aarhus University, Denmark\\
  \einstsep{}\texttt{peyman@madalgo.au.dk}
\and
  \instsep{2}University of Connecticut, USA\\
  \einstsep{}\texttt{don.r.sheehy@gmail.com}
\and
  \instsep{3}Institut f\"ur Informatik, Freie Universit\"at Berlin, Germany\\
  \einstsep{}\texttt{yannik.stein@fu-berlin.de}
}

\begin{document}

\maketitle

\begin{abstract}
Let $P$ be a set of $n$ points in $d$-dimensions. The simplicial depth,
$\sigma_P(q)$ of a point $q$ is the number of $d$-simplices with vertices in $P$
that contain $q$ in their convex hulls.  The simplicial depth is a notion of
data depth with many applications in robust statistics and computational
geometry.  Computing the simplicial depth of a point is known to be a
challenging problem.  The trivial solution requires $O(n^{d+1})$ time whereas it
is generally believed that one cannot do better than $O(n^{d-1})$.

In this paper, we consider approximation algorithms for computing the simplicial
depth of a point.  For $d=2$, we present a new data structure that can
approximate the simplicial depth in polylogarithmic time, using polylogarithmic
query time.  In 3D, we can approximate the simplicial depth of a given point in
near-linear time, which is clearly optimal up to polylogarithmic factors.  For
higher dimensions, we consider two approximation algorithms with different worst-case
scenarios.  By combining these approaches, we compute a
$(1+\eps)$-approximation of the simplicial depth in time $\hdtime$ ignoring
polylogarithmic factor.
All of these algorithms are Monte Carlo algorithms.
Furthermore, we present a simple strategy to compute the simplicial depth exactly in 
$O(n^d \log n)$ time, which provides the first improvement over the trivial
$O(n^{d+1})$ time algorithm for $d>4$.
Finally, we show that computing the simplicial depth
exactly is \#P-complete and W[1]-hard if the dimension is part of the input.
\end{abstract}

\section{Introduction}
Let $P \subset \R^d$ be a point set and $q \in \R^d$ be a point.  The \emph{simplicial
  depth}~\cite{RousseeuwRu1996} $\sigma_P(q)$ of $q$ with respect to $P$ is the
number of subsets $P' \subseteq P$, $|P'|=d+1$, that contain $q$ in their convex
hull (see also~\cite{burrsimplicial} for an alternate definition).  
This is one of the important definitions of data depth and has generated interest in both robust statistics and computational geometry
since its introduction.  Designing efficient algorithms to compute (or
approximate) the simplicial depth of a point remains an intriguing task
in this area.

Other notions of depth include halfspace (a.k.a. Tukey) depth, Oja depth, regression depth
and convex hull peeling depth~\cite{aloupissurvey}.
Among them, the one most relevant to our techniques is the Tukey depth: given a set
$P$ of $n$ points in $\R^d$, the Tukey depth, $\tau_P(q)$ of $q$ with respect to $P$
is the minium number of points contained in a halfspace that also contains $q$.

\subparagraph{Previous and Related Results.}
Computing the simplicial depth of a single point in 2D was
considered even before its formal definition~\cite{KhullerJo1990} almost three decades ago,
perhaps because it translates into an ``intuitive'' problem of counting the
number of triangles containing a given point.
In fact, at least three independent papers study this problem in 2D and show how
to compute the simplicial depth in $O(n\log n)$
time~\cite{GilStWi1992,KhullerJo1990,RousseeuwRu1996}.
This running time is optimal~\cite{AloupisCoGoSoTo2002}.
In 2003, Burr et al.~\cite{burrsimplicial} presented an alternate definition for the simplicial
depth to overcome some unpleasant behaviors that emerge when dealing with degeneracies. 
Since we will be dealing with approximations, we will assume general position and thus avoid
issues with degeneracy. 
In 3D, the first non-trivial result offered the bound of $O(n^2)$~\cite{GilStWi1992} but
it was flawed; fortunately, the running time of $O(n^2)$ could still be obtained with proper
modifications~\cite{ChengOu2001}.
The same authors presented an algorithm with running time of $O(n^4)$ in 4D.
For dimensions beyond $4$ there seems to be no
significant improvements over the trivial $O(n^{d+1})$ brute-force solution.
Furthermore, it is natural to  conjecture that computing the simplicial depth should require
$\Omega(n^{d-1})$ time:
given a set $P$ of $n$ points, it is generally conjectured that detecting whether or not $d+1$ points lie
on a hyperplane requires $\Omega(n^{d})$ time~\cite{ericksonLB} and this conjecture would imply that detecting whether $d$ points of $P$ and a fixed point
$q$ lie on a hyperplane should require $\Omega(n^{d-1})$ time. 
This is one motivation to consider the approximate version of the problem.
In fact, Burr et al.~\cite{burrsimplicial} have already expressed interest in computing an 
approximation to the simplicial depth and they propose a potential approach, although without any
worst-case analysis~\cite{burrtechnical}. In 2007, Bagchi et
al.~\cite{bagchi2007deterministic} presented a data structure for the
two-dimensional case: using $O(n\, \text{polylog}\, n)$ preprocessing time, they
can additively $\eps$-approximate the simplicial depth of a given query point in
$O(1)$ time.

Here, we  only consider relative approximation; additive approximation (with
additive error of $\varepsilon n^{d+1}$) can be obtained using
$\varepsilon$-nets and $\varepsilon$-\setlength{\hspace}{\widthof{ }*\real{-1}}
approximations (see~\cite{Chazelle2000,bagchi2007deterministic} for more details).

Another motivation for computing a relative approximation comes from applications in
outlier removal. Intuitively, statistical depth measures how deep a point is embedded
in the data cloud with outliers corresponding to points with small values of depth. 
In such applications, if a small relative error of $(1+\varepsilon)$ is tolerable, then
faster outlier removal can be possible using approximations.

There are several notions of data depth for which approximation and related
computational problems have been considered.
Aronov and Har-Peled~\cite{AronovHa2008} describe general techniques
to attack approximation problems related to various notions of depth including finding an approximately
deepest point in an arrangement of pseudodisks, approximating the depth of a query point in 
an arrangement of pseudodisks, approximate halfspace range counting and an approximate version of linear 
programming with violations which both can be formulated to depth-related problems in an arrangement of halfspaces.

Approximate halfspace range counting received most of the 
attention~\cite{Afshani.Chan.SOCG07,AHZ.UB.CGTA,KaplanSh2006} but this also renewed 
interest on the general study of relative approximations~\cite{Aronov.et.al.SOCG07}.
Continuing this line of research Afshani and Chan presented data structures to 
approximate the depth of a query point in an input set of points for Tukey depth
in 3D and regression depth in 2D~\cite{Afshani.Chan.SOCG07}. 

\paragraph*{Our Results.} 
In Sections~\ref{sec:ldim} and \ref{sec:3d}, we consider the simplicial depth
problem in $2$ and $3$ dimensions. For $d=2$, we present a data structure of
size $\tO(n)$\footnote{The $\tO(\cdot)$ notation hides
a constant number of polylogarithmic factors of $n$, e.g., $\log n = \tO(1)$.} with
$\tO(n)$ preprocessing that returns the relative $\eps$-approximation of
the simplicial depth of a query point $q$ with high probability in $\tO(1)$
time, where $\eps > 0$ is an arbitrary constant.
In Section~\ref{sec:hdim}, we consider the simplicial depth problem in arbitrary
but fixed dimensions. We present two algorithms that each compute a
$(1+\eps)$-approximation, however with different worst-case scenarios.
A combination of these strategies gives an algorithm that returns a
$(1+\eps)$-approximation of the simplicial depth with high probability in
$\hdtime$ time.
Finally, we show in Section~\ref{sec:complexity} that computing
the simplicial depth becomes \#P-complete and W[1]-hard with respect to the
parameter $d$ if the dimension is part of the input.

\paragraph*{Technical Difficulties.}
One standard technique to approximate various geometric measures is
the use of uniform random samples combined with Chernoff type inequalities.
Often uniform random sampling enables us to approximate the depth with high
probability if the depth lies in a certain range.
This property is exploited by building a hierarchy of random samples that cover all possible
ranges of data depth
(see~\cite{Afshani.Chan.SOCG07,AronovHa2008,Chan.Range.Reporting.SJC00,KaplanSh2006}).
However, these existing techniques seem insufficient to approximate the
simplicial depth.
One particular troubling situation is depicted in Figure~\ref{fig:badsample}:
in this configuration no high probability bound can be achieved
for the simplicial depth of $q$ in a uniform random sample despite the fact that
$q$ is far from being an outlier (it has Tukey depth $\Theta(n^{1/3})$ and
simplicial depth $\Theta(n^2)$).
This poses serious problems for all the previous techniques including
the general techniques of Aronov and Har-Peled~\cite{AronovHa2008}, 
Afshani and Chan~\cite{Afshani.Chan.SOCG07} and Kaplan and 
Sharir~\cite{KaplanSh2006}.
In fact, it can be seen that the definition of simplicial depth
prevents us from using any technique which depends on a Chernoff-type
inequality.

\begin{figure}[thpb]
	\begin{center}
		\includegraphics[scale=0.5]{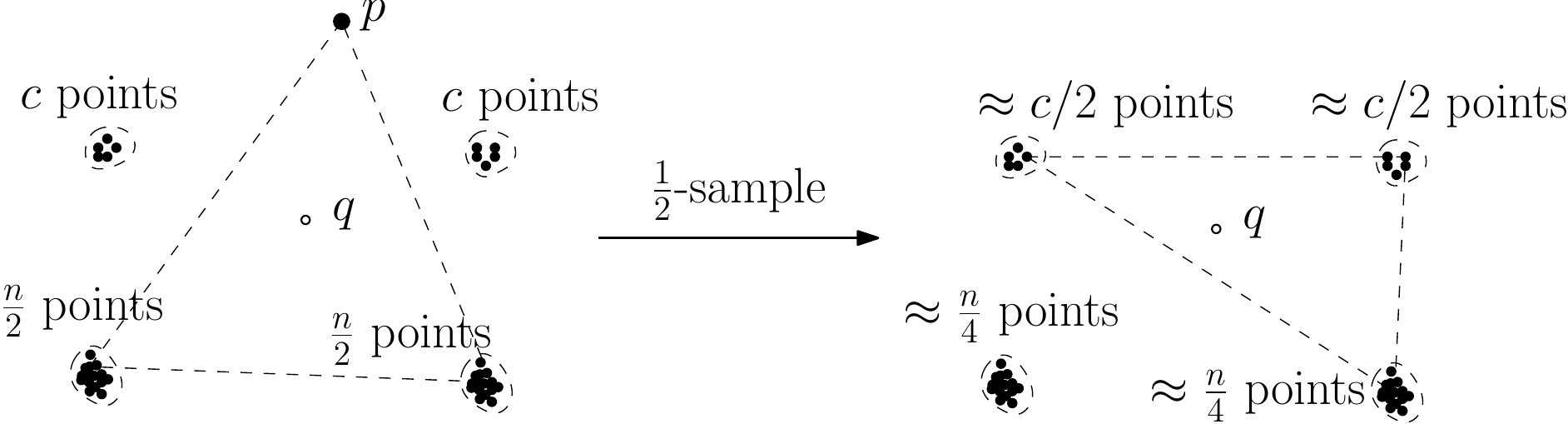}
	\end{center}
    \caption{(a) The point $q$ has simplicial depth $\Theta(n^2)$ for $c=O(1)$.
      (b) The random sample misses point $p$ and the $n^2/4$ triangles that shared $p$ as a vertex; the simplicial
	depth of $q$ is now $O(n^{4/3})$.
	}
	\label{fig:badsample}
\end{figure}

\section{Structural Theorems and Preliminaries}\label{sec:struct}
In this section, we provide non-algorithmic results and review properties of
simplicial depth which can be of independent interest. In the following sections,
these results are used extensively to approximate the simplicial depth.

\subsection{Notation}
Throughout this paper, we denote the input set of $n$ points in
$\R^d$ (for a constant $d$) by $P$ and the query point by $q$.
For a given set of points $P$, the Tukey depth of a point $q$
with respect to $P$, denoted by $\tau_P(q)$, is the minimum number of points of
$P$ that is contained in a halfspace $h$ with $q \in h$.
We will always assume the set $P\cup \{q\}$ is in general position.
We denote with $\Delta_P(q)$ all subsets $P' \subseteq P$, $|P'| = d+1$, that
contain $q$ in their convex hull.
For a point $p \in \R^d$, let $\Delta_{p;P}(q) \subseteq \Delta_P(q)$ be the set of
simplices formed by $p$ and $d$ points in $P$ that contain $q$.
We define the weight function $\omega_{q;P}(p) = |\Delta_{p;P}(q)|$
(we might omit $P$ in the subscript if there is no danger of ambiguity).
Note that we have $\sum_{p\in P} \omega_{q;P}(p) = (d+1)\sigma_P(q)$.
Given any two points $v$ and $u$, we denote the line passing through them
by $\overline{uv}$ and the ray from $v$ and through $u$ by
$\overrightarrow{vu}$.
We call a set $S$ an $\alpha$-sample of a set $U$ if every
element of $U$ is chosen uniformly and randomly in $S$ with 
probability $\alpha$.
We will say an event $X$ happens with high probability if 
$\Pr [ X ] \ge 1 - n^{-c} $ in which $c$ is some (large) constant.
\subsection{Bounding the Simplicial Depth with Tukey Depth}
The main result of this section are tight asymptotic bounds between $\sigma_P(q)$ and $\tau_P(q)$.
These results will be useful later on, but we believe they might be of independent interest as well.
We start with the following easy lemma.

\begin{lemma}\label{lem:change}
  Let $P\subset \R^d$ be a point set and $q \in \R^d$.
	Let $S\subseteq P$, $|S|=d+1$, be a subset with $q\in \conv(S)$.
    For every point $p \in P$, there exists a unique vertex $v$ of $S$ such that
    the simplex $S\cup \left\{ p \right\}\setminus \left\{ v \right\}$ 
	contains $q$. 
\end{lemma}
\begin{proof}
	Consider the ray $r$ starting from $q$ in direction $q-p$.
    Since $q$ is inside $\conv(S)$, $r$ intersects a facet of $\conv(S)$.
	An easy calculation reveals the vertex $v$ opposite to this facet is the vertex claimed in the lemma.
    Furthermore, $v$ is unique since we assume general position.
\end{proof}

We call the operation used in the above lemma \emph{swapping}.
The main result of this section is the following.

\begin{lemma}\label{lem:SH}
	For any point set ${P}\subset \R^d$ and any $q\in \R^d$,
	$\sigma_{P}(q) = \Omega(|{P}|\tau_{P}^d(q))$ and $\sigma_{P}(q) = O(|{P}|^d\tau_{P}(q))$. 
	Furthermore, these bounds are tight.
\end{lemma}
\begin{proof}
	If $\tau_{P}(q) = 0$, then $q$ is outside the convex hull of ${P}$ and there is 
	nothing left to prove, so assume otherwise.
	By \Caratheodory's theorem, there exists a simplex $\Delta_1$ formed by $d+1$ points of 
	${P}$ which contains $q$.
	Removing $\Delta_1$ from ${P}$ reduces the Tukey depth of $q$ by at most $d$ and 
	by repeating this operation we can find $m$ disjoint subsets 
    ${\Delta}_1, \dots, {\Delta}_m \subset {P}$ where $m \ge \frac{\tau_{P}(q)}{d}$.
	Let ${A}= \bigcup {\Delta}_i$.
	Using the Colorful \Caratheodory Theorem, and 
    following the exact same technique as B\'ar\'any~\cite{Barany1982}, it follows that
	$\sigma_{A}(q) = \Omega(m^{d+1})=\Omega(\tau_{P}^{d+1}(q))$.
    We assume $|{A}| \le  |{P}|/2$, otherwise ${\Delta}_A(q)$ already contains
	$\Omega(|{P}|^{d+1})$ simplices and there is nothing left to prove.
	Using Lemma~\ref{lem:change}, for any 
	$p \in {P}\setminus {A}$ and for every simplex $\Delta \in \Delta_A(q)$,
	we can create another simplex $\Delta'$ by adding $p$ and removing some other point.
	This way, we can create $\sigma_A(q)$ simplices for every point $p$.
	However, some of these simplices could be identical. 
	The main observation is that any simplex $\Delta'$, can be obtained through at most
	$|A|$ different ways.
	Thus, it follows that there are at least 
    $\Omega(|\Delta_A(q)|/|{A}|)= \Omega(\tau_{P}^d(q))$
    distinct simplices $\Delta_{p, A}$ with $p$ as a vertex and no
    other point from $P \setminus A$.
	For two points $p$ and $p'$ in ${P}\setminus {A}$, the sets $\Delta_{p,
    A}$ and $\Delta_{p', A}$ are disjoint.
	Hence, in total we have produced $\Omega(|{P}|\tau_{P}^d(q))$ distinct simplices.
	
	To prove the upper bound, consider a halfspace $h$ that passes through $q$ 
	and contains $\tau_{P}(q)$ points. 
	Every simplex containing $q$ must have at least one point from 
	$h\cap P$. The maximum number of such simplices is at most $|{P}|^d\tau_{P}(q)$.

	To demonstrate the tightness of the upper bound, consider a simplex $\Delta$ and a point
	$q$ inside it.
	Replace one vertex of the simplex with a cluster of $m$ points placed closely to
	each other and replace all the remaining vertices with clusters of $n$ points.
	The resulting point set ${P}_1$ will contain $\Theta(n)$ points with 
	$\sigma_{{P}_1}(q) = m n^d$ and $\tau_{{P}_1}(q) = m$.
	The tightness of the lower bound is realized by a very similar construction but using
	clusters of size $m$ at every vertex except one, and using a cluster of size $n$ at 
	the remaining vertex.
	The resulting point set ${P}_2$ will contain $\Theta(n)$ points with 
	$\sigma_{{P}_2}(q) = m^d n$.
\end{proof}

\subsection{Properties of Random Samples}
A big obstacle to approximating the simplicial depth is that
it is not easy to estimate the simplicial depth from a random sample:
the probability that a simplex $\Delta$ survives in an $\alpha$-sample of the point set
is exactly $\alpha^{d+1}$ but these probabilities can be highly dependent for different simplices.
Because of this, we need proper tools to deal with such dependence. 
One such tool is Azuma's inequality.

\begin{azuma}\label{thm:azuma}
	Suppose $\{X_k\}$ is a martingale with the property that
	$|X_{i}-X_{i-1}| \le c_i$. 
	Then
    $$\Pr[ |X_n - X_0| \ge t ] \le e^{-\frac{t^2}{2\sum_{k=1}^n c_k^2}}.$$
\end{azuma}

The following lemma is our main tool for estimating the simplicial depth from random samples.

\begin{lemma}\label{lem:large.HS}\label{lem:RS}
	Let ${P} \subset \R^d$ be a point set of size $n$ and $q\in \R^d$ an arbitrary
	point.
	There exists a universal constant $C$ such that for any parameter $\eps > 0$ the following
	holds.
    Set $M = C^{-1} \varepsilon^2 \sigma_{P}(q) / \log n$ and pick
    $P_\lar$ as a subset of $P$ that includes all points $p$ with $\omega_{q;P}(p) \ge M$.	
	Build a sample $S \subseteq P$ by adding a $1/2$-sample of $P \setminus P_\lar$ to $P_\lar$.
    Then, the event $|\E(\sigma_{S}(q))- \sigma_{S}(q)| \ge \varepsilon \sigma_P(q)$ holds
	with high probability.
\end{lemma}
\begin{proof}
	Let $P' = P\setminus P_\lar$ and let $S'$ be a $1/2$-sample of $P'$.
	By construction, we have $S = S' \cup P_\lar$.
	Let $p_1, \dots, p_{n'}$ be an ordering of the points of ${P'}$.
	We build a martingale by revealing presence or absence of points
	of $P'$ in $S'$ in this order.
	Define $x_i = 1$ if $p_i$ is sampled in $S'$ and 0 otherwise.
	Let $X_i$ be the random variable corresponding to the expected value
	of $\sigma_{S}(q)$ in which the values of $x_1, \dots, x_i$ have been revealed. That is, 
	the expectation is taken over $x_{i+1}, \dots, x_{n'}$.
	According to this definition, $X_{n'}$ is equal to $\sigma_{S}(q)$, since
    we have revealed all the points in our sample, while 
	$X_0$ is equal to $\E(\sigma_{S}(q))$, since we have revealed nothing.

	The sequence $X_0, \dots, X_n$ has the martingale property:
    \begin{align*}
      \E(X_{i+1} | X_1, \cdots, X_i) = \E(X_{i+1} | X_i)
    \end{align*}
    and furthermore, the difference between
	$X_{i+1}$ and $X_{i}$ is the knowledge of $x_{i+1}$. However,
    a simple calculation reveals that the contribution of a simplex with
    $p_{i+1}$
    as a vertex is exactly the same in both $\E(X_{i+1} | X_i)$ and $\E(X_i)$. Thus,
	$\E(X_{i+1} | X_i) = \E(X_i)$.

    Let $c_i=| X_{i} - X_{i-1}|$ as in Azuma's inequality.
    We now show that $c_i$ is at most $\omega_{q;P}(p_{i}) < M$.
    If $p_{i}$ is not sampled, then any simplex with $p_{i}$ as a vertex
    will not survive and thus their contribution to the expected value of
    $\sigma_{S}(q)$ will be zero, a decrease of at most $\omega_{q;P}(p_{i})$ in the expected value.
    So assume $p_{i}$ is part of the sample, and consider a simplex $\Delta$ that contains $q$ and is composed
    of $p_{i}$, $t$ points $p'_1, \cdots, p'_t \in \left\{  p_{i+1}, \cdots,
      p_{n'}\right\}$, and $d-t$ 
    point $p''_1, \cdots, p''_{d-t} \in \left\{ p_1, \cdots, p_{i-1} \right\}$.
    If any of the points $p''_1, \cdots, p''_{d-t}$ have not been sampled, then the contribution of
    $\Delta$ to the expected value of $\sigma_{S}(q)$ is zero so assume we have revealed
    that all these points have been sampled in $S$.
    The contribution of $\Delta$ to the expected value of $\sigma_{S}(q)$ before revealing
    that $x_{i} = 1$ was exactly $2^{-(t+1)}$, equal to the probability that we sample all the
    points $p'_1, \cdots, p'_t$ and $p_{i}$. After revealing $x_{i} =1$, this contribution
    increases to $2^{-t}$.
    Clearly, over all simplices with $p_{i}$ as a vertex, this increase is at
    most $\omega_{q;P}(p_{i})$.
    Since $p_{i} \not \in P_\lar$, the magnitude of the change is at most $M$ in
    both cases.

    As discussed,
	$X_{n'} = \sigma_{S}(q)$ and $X_0 = \E(\sigma_{S}(q))$.
	Note that $X_0$ is not a random variable and that we have 
    $\sum_{i=1}^{n'} c_i \le \sum_{i=1}^{n'-1} \omega_{q;P}(p_{i+1}) = O(\sigma_P(q))$.
	By Azuma's inequality we have
    \begin{multline*}
      \Pr[ |X_{n'} - X_0 | \ge \varepsilon \sigma_P(q) ] \le
      e^{-\frac{\varepsilon^2 \sigma^2_P(q)}{\sum c_i^2}} \le
      e^{-\frac{\varepsilon^2 \sigma^2_P(q)}{M\sum c_i}} \le  
      e^{-\Omega\left(\frac{\varepsilon^2 \sigma_P(q)}{M} \right)}\\ =
      e^{-\Omega\left(\frac{\varepsilon^2 \sigma_P(q) \log n}{C^{-1} \varepsilon^2
          \sigma_{P}(q)} \right)}.
    \end{multline*}
	The lemma follows by picking $C$ large enough.
\end{proof}

The above lemma can be used to reduce the problem of computing $\sigma_P(q)$ to
computing the simplicial depth of $q$ with the respect to a set $S$ of roughly half
the size of $P$. Furthermore, the value of $\E(\sigma_S(q))$ is directly tied to the value of
$\sigma_P(q)$: any simplex $\Delta \in \Delta_P(q)$ that contains $t$ points from
the set $P'$ and $d+1 - t$ points from $P_\lar$ contributes exactly $2^{-t}$ to
$\E(\sigma_S(q))$.

\ignore{
\paragraph{Remarks.}
There are many articles that deal with concentration of 
functions of independent variables and the results obtained in
many of them can possibly be applied to our problem as well. 
For instance, we can formulate our problem easily in Van Vu's polynomial
model~\cite{VanVu} although the final bound will be slightly worse than what we have
obtained above. 
We can also use Talagrand's inequality~\cite{Talagrand} and obtain an equivalent result
to that of Lemma~\ref{lem:RS}. 
These imply the result of Lemma~\ref{lem:RS} might not be optimal in terms
of dependence on $\gamma$ although the situation depicted in 
Figure~\ref{fig:problem} implies we cannot hope to get a better bound than
$e^{-\Theta(\gamma r)}$.
}

\section{Approximating the Simplicial Depth in \texorpdfstring{$\mathbf{2}$}{2}
  Dimensions}
\label{sec:ldim}
The main result of this section is a data structure of near-linear size that can 
answer approximate simplicial depth queries in polylogarithmic time.
Later, it will be used to get an almost-optimal algorithm for approximating the
simplicial depth in 3D.
The problem can be stated as a triangle counting problem:
given a set of $n$ points ${P}$ in the plane, build a data structure 
capable of approximating
the number of triangles formed by points of ${P}$ that contain a query point $q$.

Let $S$ be a $\frac{1}{2}$-sample from $P$.
Consider the simple and easy case when $P_\lar$ is empty. Then, we have
$\E(\sigma_S(q)) = \sigma_P(q)/8$ by linearity of expectation and the observation
that any triangle made by points of $P$ containing $q$
survives with probability $1/8$. By Lemma~\ref{lem:RS}, we can conclude that
a recursively computed approximation for $\sigma_S(q)$ is with high probability
very close to $\sigma_P(q)/8$.
If $P_\lar$ is not empty, then it can only contain polylogarithmically many points
(as each point contributes a significant amount to the simplicial depth) and thus
we need to keep track of a ``few'' points.
The biggest challenge, however, is finding the subset $P_\lar$.
This is done with the following lemma.
    Unfortunately its proof does not seem to be easy and in fact it requires overcoming
    many technical steps and combining shallow cuttings with various observations
    regarding the geometry of planar points. We ultimately reduce the problem
    to instances of orthogonal range reporting problem in eight(!) dimensional space, which fortunately can
    be solved with $\tilde{O}(n)$ space and $\tilde O(1)$ query time.
    The proof is given in Section~\ref{sec:findlarge}.

\begin{restatable}{lemma}{rlemfindlarge}
    \label{lem:findlarge}
	Let  $P$ be set of $n$ points and let 
    $S'_0 = {P}, S'_1, \dots, S'_r$ in which
    $S'_{i+1}$ is a $\frac{1}{2}$-sample from $S'_i$.
    There exists a data structure of size $\tilde{O}(n)$, such that the following holds with high probability.
    Given $j$ and a query point $q$, define $M = C^{-1} \varepsilon^2 \sigma_{S'_j}(q) / \log n$, where
    $C$ is a constant and $\eps > 0$ is a fixed parameter.
    If $\tau_{S_j}(q) = \Omega(\varepsilon^{-2} \log^3 n)$, then
    the data structure can find the set $P_\lar \subset S_j$ that contains all
    the points $p$ with $\omega_{q;S_j}(p) \ge M$ in $\tilde{O}(1)$ time.
	The data structure can be built in $\tO(n)$ expected time. 

    Also, the data structure defines $O(n)$ canonical halfplanes, such that 
    $P_\lar$ lies inside a halfplane $h(q)$ that contains $q$ and $\tilde{O}(\tau_{S_j}(q))$ points of $S_j$.
    $h(q)$ only depends on $q$ and not $j$.
\end{restatable}

We also need the following lemma.
For the proof, see Section~\ref{sec:calc}.
\begin{restatable}{lemma}{lemcalc}\label{lem:calc}
    Given a parameter $\varepsilon > 0$, 
    we can store a set $S$ of $n$ points in a data structure of size $\tilde{O}(n)$ 
    that can answer the following queries. 
	Given a query point $q$, a halfplane $h$ 
	with $q$ on its boundary containing
    $\tilde{O}(\tau_S(q))$ points of $S$, and a subset of $R\subset S \cap h$, we
    approximate the total number of triangles that contain $q$ and include at least one point 
    from $R$ in $\tilde{O}(\varepsilon^{-2} |R|^2)$ time and with additive error of 
	at most $\varepsilon \sigma_S(q)$.
\end{restatable}

In the rest of this section, we outline our solution to approximate the simplicial depth of a query point.
Consider a series of random samples 
${S'}_0 = {P}, {S'}_1, \dots, {S'}_r$ in which $S'_{i+1}$ is a $1/2$-sample from $S'_i$
and $|S_r| = \tilde{O}(1)$.
We store each sample $S'_i$ in the data structure from Lemma~\ref{lem:calc}.
Furthermore, we store the sampling sequence
sampling sequence $S'_0, \cdots, S'_r$ in the data structure from Lemma~\ref{lem:findlarge}.

We use a recursive approach where Lemmata~\ref{lem:RS},~\ref{lem:findlarge},
and \ref{lem:calc} are our bread and butter: during step $i$ of the query
algorithm, 
we are given a set $S\ti_\lar \subset R(q)\cap P$ containing $\tilde{O}(1)$ points, such that
$S\ti_\lar \cap S'_i = \emptyset$.
The goal is to compute $\sigma_{S_i}(q)$ where $S_i = S'_i \cup S_\lar$.
Initially, $S^{(0)}_\lar = \emptyset$ and $S'_0 = P$.

Our strategy will be to recursively compute the simplicial depth.
We will assume our recursion returns a relative $(1+\delta)$-approximation of
$\sigma_{S_{i+1}}(q)$. 
By using Lemmata~\ref{lem:RS},~\ref{lem:findlarge}, and \ref{lem:calc} with parameter $\varepsilon$, 
we can return a value that is a relative $(1+\delta + O(\varepsilon))$-approximation of 
$\sigma_{S_i}(q)$. 
We set $\varepsilon = 1/\delta$ such that at the top of the recursion we end up
with a relative $(1+O(\delta))$-approximation. Now we present the details.

\begin{figure}[h]
    \centering
    \includegraphics[scale=0.5]{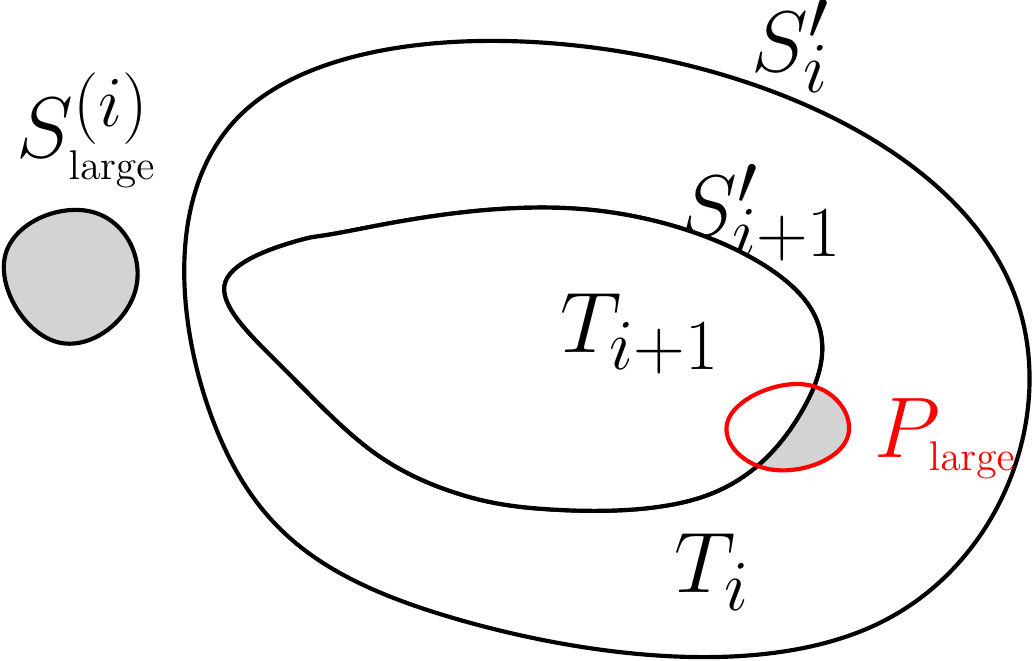}
    \caption{$S'_{i+1}$ is a $1/2$-sample of $S'_i$. 
    $P_\lar \subset S'_i$ is obtained through Lemma~\ref{lem:findlarge} and it is shown in red.
    $S_i = S\ti_\lar \cup S'_i$, $T_{i+1} = S'_{i+1}\setminus P_\lar$, $T_{i} = S'_{i}\setminus P_\lar$, and
     the greyed areas represent $S\tii_\lar$.}
    \label{fig:venn}
\end{figure}
\paragraph*{Approximating $\sigma_{S_i}(q)$.}
Let $\varepsilon = \delta/\log n$.
First, observe that if $\tau_{S_i}(q) = \tilde{O}(1)$, then we can directly approximate 
$\sigma_{S_j}(q)$ using Lemma~\ref{lem:calc} by setting $R = R(q) \cap S_j$ in
Lemma~\ref{lem:calc}.
In the rest of this proof, we assume that $\tau_{S_i}(q) = \Omega(\varepsilon^{-2} \log^3 n)$.

See Figure~\ref{fig:venn} for a Venn diagram  of the various subsets involved
here.  We set $M=C^{-1} \varepsilon^2 \sigma_{S'_i}(q) / \log n$ and using
Lemma~\ref{lem:findlarge}, we find the set $P_\lar \subset S'_i$. Let
$T_{i+1} = S'_{i+1}\setminus P_\lar$, $T_{i} = S'_{i}\setminus P_\lar$, and
$S\tii_\lar = (P_\lar \setminus S'_{i+1}) \cup S\ti_\lar$.  Clearly,
$S\tii_\lar \cap S_{i+1} = \emptyset$, and $S\tii_\lar \subset R(q) \cap P$ so
we can recurse. Assume we obtain a value $Y$ that is a relative $(1+\delta)$ approximation
of $\sigma_{S_{i+1}}(q)$.  Note that $T_{i+1}$ is a $1/2$-sample of $T_i$.
From the definition of $P_\lar$ and by Lemma~\ref{lem:RS}, we can
conclude that with high probability
\begin{align}
    |\E[\sigma_{T_{i+1} \cup P_\lar}(q)] - \sigma_{T_{i+1} \cup P_\lar}(q)| \le \varepsilon \sigma_{S'_{i}}(q) \label{eq:high}.
\end{align}
By Lemma~\ref{lem:calc}, we can approximate the number of triangles, $m$, that
contain at least one point from $S\ti_\lar$ with additive error
$\varepsilon \tau_{S_{i+1}}(q)$. This combined with $Y$ gives an approximate
value of $\sigma_{T_{i+1} \cup P_\lar}(q)$ with additive error
$(O(\varepsilon) + \delta) \tau_{S_{i+1}}(q)$.  Combined with (\ref{eq:high}),
this gives an estimate for $\E[\sigma_{T_{i+1} \cup P_\lar}(q)]$ with additive
error $(O(\varepsilon) + \delta) \tau_{S_{i}}(q)$.  Let $n_i$, $0 \le i \le
3$, be the number of triangles containing $q$ that have $i$ points from the set
$P_\lar$ and $3-i$ points from $T_i$.  Clearly, $\sigma_{S'_i}(q) = n_0 + n_1 +
n_2 + n_3$ and $\E[\sigma_{T_{i+1} \cup P_\lar}(q)] = n_0/8 + n_1/4 + n_2/2 +
n_3$.  Using Lemma~\ref{lem:calc}, we approximate $n_1 + n_2 + n_3$ with
additive error $\varepsilon \tau_{S_{i}}(q)$. 
However, $n_2$ and $n_3$ are negligible compared to our error margins:
by Lemma~\ref{lem:SH}, $\sigma_{S_i}(q) = \Omega(|S_i|\tau^2_{S_i}(q))$
but $n_2$ can be at most $|P_\lar|^2 |S_i|$ and $n_3$ can be at most $|P_\lar|^3$.
Since $\tau_{S_i}(q) = \Omega(\varepsilon^{-2} \log^3 n)$, it follows that
we can ignore $n_2$ and $n_3$ in our calculations and that Lemma~\ref{lem:calc} 
reveals an approximation of $n_1$.
With these observations, we can 
obtain an approximation $\sigma_{S'_i}(q)$  with additive error
$(O(\varepsilon) + \delta) \tau_{S_{i}}(q)$.  Combining this with an
approximation of $m$ yields an approximation of $\sigma_{S_i}(q)$ with additive
error $(\delta + O(\varepsilon))\sigma_{S_i}(q)$.
Observe that the error factor in our additive term has worsened from $\delta$ 
(regarding $Y$ and $\sigma_{S_{i+1}}(q)$) to $\delta + O(\varepsilon)$
(regarding $\sigma_{S_{i}}(q)$).
However, we have $\varepsilon = \delta/\log n$ and there are at most $O(\log n)$ recursion steps.
Thus at the top of the recursion (that is for $\sigma_P(q)$), we obtain a
$(1+O(\delta))$ approximation factor, as claimed.

\begin{theorem}\label{thm:2d}
	It is possible to preprocess a point set ${P}\subset \R^2$ of
	$n$ points in $\tO(n)$ expected time
	using $\tO(n)$ space such that, given a query point $q$,
	a relative $\varepsilon$-approximation for the simplicial depth of $q$ 
	can be found in 
	$\tO(1)$ expected time for any arbitrary fixed constant $\varepsilon > 0$
	with high probability.
\end{theorem}

\section{Proof of Lemma~\ref{lem:findlarge}}~\label{sec:findlarge}

\begin{restatable}{lemma}{lemob}\label{lem:ob1}
    Let ${P} \subset \R^2$ be a set of $n$ points, $q \in \R^2$ an 
    arbitrary point and $h$
    a halfplane containing at most $C'\tau_{{P}}(q)$ points
    with $q$ at its boundary.
    Consider the line $\overline{pq}$ for a point $p \in h$ and assume
    it partitions ${P}$ into two sets of sizes $n_1$ and $n_2$.
    If $n_1, n_2 \ge 2C'\varepsilon^{-1}\tau_{{P}}(q)$ then 
	$n_1 n_2$ is a relative (1+$\varepsilon)$-approximation of $\omega_{q;P}(p)$.
\end{restatable}
\begin{proof}
    The number of triangles containing $q$ and involving $p$ is at least
    $(n_1 - m)(n_2 -m)$ and at most $n_1 n_2$ where $m \le  C'\tau_{{P}}(q)$ 
	is the number of points in $h$.
    The lemma follows by a simple calculation and observing that $n_1 n_2$ is minimized
    when $n_1$ or $n_2$ is $2C'\varepsilon^{-1}\tau_{{P}}(q)$.
\end{proof}

We also use shallow cuttings which we state in two dimensions.
For a set of lines $H$ in the plane, the \emph{level} of a point $p$ is the number
of points that pass below $p$. 
Given an integer $k$, the $k$-level of $H$ is the closure of all the points that have
level exactly $k$.
The $(\le k)$-level is defined as the closure of all the points with level at most $k$.
\begin{theorem}\label{thm:shallow}
	Let $H$ be a set of $n$ lines in the plane and $k$ be a given parameter $1 \le k < n/2$.
	We can find a convex polygonal chain $C$ of size $O(n/k)$ such that it lies above
	the $k$-level of $H$, the level of every vertex of $C$ is $O(k)$.
	The cutting can be constructed in $O(n\log n)$ time. 
\end{theorem}
\begin{proof}
	Matou\v sek~\cite{Matousek.reporting.points} proved that one can cover the $(\le k)$-level
	with $O(n/k)$ triangles such that there are at most $O(k)$ lines passing below
	each triangle. 
	Chan~\cite{Chan.Range.Reporting.SJC00} observed that we can consider the convex hull of the triangles; 
	the number of lines passing below the convex hull is only increasing by $O(k)$.

	Ramos~\cite{Ramos.SOCG99} offered a randomized $O(n\log n)$ construction in 1999
	and recently Chan and Tsakalidis have shown the same running time can be
	achieved with a deterministic algorithm~\cite{chanshallow}.
\end{proof}

We will also be working with point-line duality in the plane.
This transformation, maps a line $\ell$ passing below (reps.\ above) a point $p$
to point $\dual{\ell}$ that lies below (resp.\ above) the line $\dual{p}$.

\rlemfindlarge*
\begin{proof}
    We will describe the data structures incrementally.
    Consider the set $\overline{{P}}$ of $n$ lines dual to $P$.
    Let $k_i = 2^i$, $1 \le i < \log n$. 
	First, by \Mat's shallow cutting theorem (Theorem~\ref{thm:shallow}), for each $k_i$, we build a shallow
    cutting $L_i$ for the $(\le k_i)$-level of $\dual{P}$, that is a convex
    polygonal chain of size $O(n/k_i)$ that lies between $k_i$-level and $O(k_i)$-level of
    $\overline{P}$. 
    Each vertex $v$ of the polygonal chain $L_i$ defines a canonical halfplane in primal space 
    (the region below $\dual{v}$).
    For each $S_j$, the subset of $S_j$ that lies inside a canonical halfplane is called a \emph{canonical set}.
    We also do the same for the $(\ge k_i)$-level (to obtain cuttings $U_i$).
	By Theorem~\ref{thm:shallow}, each chain $L_i$ has $O(n/k_i)$ vertices and thus creates $O(n/k_i)$ canonical sets.
	Furthermore, each canonical set contains $O(k_i)$ points.
    Thus, the total size of canonical sets on $P$ is $O(n\log n)$ which is also an upper bound for the total
	size of the canonical sets on each $S_i$. 
	This means, the total size of canonical sets,
    over all indices $j, 1 \le j \le \log n$ is $O(n\log^2 n)$.
    Furthermore, a standard application of the Chernoff bound yields that $L_i$ is below
	$(\le \Theta(k_i\log n/2^j))$-level and above $(\le (\Theta(k_i /(2^j \log
	n)))$-level of $S_j$, with high probability (similarly for $U_i$); let's
	call this the \emph{level property}.
    In the rest of this proof, we will build a separate data structure for each $S_j$.
    There are $O(\log n)$ different indices $j$ and thus this would only blow up the space by a $\log n$ factor.
    Let $S=S_j$ be the subset we are currently working with. 
    Given a query point $q$, the goal is to report a subset $P_\lar \subset S$ that contains all
    the points $p$ with $\omega_{q;S}(p) \ge M$.

    \begin{figure}[h]
        \centering
        \includegraphics[scale=1.0]{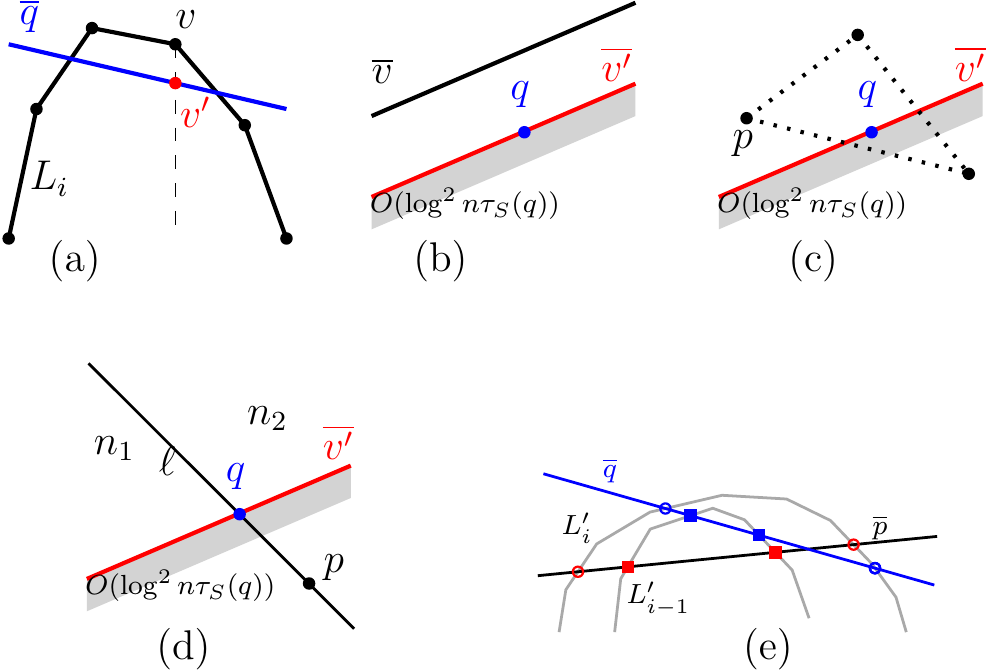}
        \caption{(a) In dual space, $\overline{q}$ passes below a vertex $v$ of a shallow cutting. 
		The set of lines below $v$ are considered a \emph{canonical set}.
		(b) In primal space, the region below $\dual{v}$ is a canonical halfplane. We
		can draw a line $\dual{v'}$ parallel to $v$ from $q$.
		(c) Any point that is above $\dual{v'}$ cannot create $M$ triangles that contain $q$ since it is forced
		to pick a point from below $\dual{v'}$ and there are only few such points.
		(d) $p$ lies below $\dual{v}$. The line $\ell$ is defined by connecting $p$ and $q$. 
	There are $n_1$ points below the line $\ell$ and $n_2$ points above it.
		(e) If a line intersects the convex chain $L'_i$, then it creates an interval on its boundary.
		We can tell if two lines intersects between chains $L'_{i-1}$ and $L'_i$ by examining the
	corresponding intervals that they define. }
        \label{fig:larged}
    \end{figure}

    The given query  point $q$, corresponds to a query line $\overline{q}$ in the dual space.
    If $i$ is the smallest index such that $\overline{q}$ intersects $L_i$ or $U_i$, then
    $k_i/2^j$ is an approximation of $\tau_{S}(q)$ up to $O(\log^2 n)$ factor, by the level property.
    W.l.o.g, assume $\overline{q}$ intersects $L_i$ and thus passes below a vertex $v \in L_i$
    (the other case when $\overline{q}$ intersects $U_i$ can be found in an analogous way and by
    building analogous data structures).
	Let $v'$ be the point on $\overline{q}$ the lies directly below $v$ (Figure~\ref{fig:larged}(a)).
    In the primal space, $v'$ corresponds to a line that goes through $q$ and has $O(\log^2 n\tau_{S}(q))$
    points below it, with high probability (by definition of $\tau_S(q)$, there are at least
	$\tau_S(q)$ points below $\dual{v'}$ as well).
    Furthermore, $v'$ is parallel to a canonical halfplane that is the dual of $v$ 
    (see Figure~\ref{fig:larged}(a,b)).
    $v'$ defines $h(q)$.

    We now claim, $P_\lar$ has to be a subset of points of $S$ that are below $\overline{v'}$.
    See Figure~\ref{fig:larged}(c).
    Consider a point $p$ that is not below $v'$; any triangle that contains $q$ and includes $p$,
    must include a point below $v'$.
    Remember that $v'$ has only $O(\log^2 n \tau_{S}(q))$ points below it. 
    This means, $\omega_{q;S}(p)=  O(|S| \log^2 n \tau_{S}(q))$.
     By Lemma~\ref{lem:SH}, we have $\sigma_S(q) = \Omega(|S| \tau^2_S(q))$.
    Note that we need to answer queries for when $\tau_S(q) = \Omega(\varepsilon^{-2} \log^3 n)$,
	namely, when $\tau_S(q) > C\varepsilon^{-2} \log^3 n$,
    so we have $\sigma_S(q) = \Omega(|S|\varepsilon^{-2}\log^3 n \tau_S(q))$.
    This means $\omega_{q,S}(p) < M$.

    Now that we have established $P_\lar$ has to be a subset of $S$ below $\dual{v'}$, we turn our attention to 
    building the proper data structures to find it.

    Consider a canonical set that contains the subset of $S$ below the line $\dual{v}$; denote this
	canonical set by $v^{\downarrow}$ and let $|v^\downarrow| = m$.
    We can build $O(m)$ canonical \emph{subgroups}, where each subgroup is a subset of $v^\downarrow$
    such that the points below any line $\dual{v'}$ parallel to $\dual{v}$ can be
    expressed as the union of at most $\log m$ canonical subgroups:
    this is done by projecting the points of $v^\downarrow$ onto a line perpendicular to $\dual{v}$,
    and then building a balanced binary tree on the resulting one-dimensional point set; each
    node of the balanced binary tree defines a canonical subgroup.
    The total size of the canonical subgroups created on $v^\downarrow$ is $O(m\log m)$.

    Remember that our goal was to find all the points $p$ below $\dual{v'}$ with $\omega_{q;S}(p) \ge M$.
    We can now use the canonical subgroups, since there $O(\log m)$ subgroups that cover all the
	points below $\dual{v'}$;
    we can query each such subgroup independently to find the subset of $P_\lar$ that lies in that subgroup;
    this would only blow up the query time by a $\log n$ factor.
    Furthermore, remember that the total size of all canonical sets was $O(n\log^2n)$ and thus the total size
    of all canonical subgroups is $O(n\log^3 n)$.
    Thus, we can afford to build  a separate data structure for each canonical subgroup.

    Consider a canonical subgroup $G \subset v^\downarrow$ that contains $g$ points.
    Our new goal is to find all  points $p \in G$  such that $\omega_{q;S}(p) \ge M$.
    We claim to approximate $\omega_{q;S}(p)$ it suffices to draw the 
    line $\ell=pq$ and then multiply the number of points of $S$ that lie at either side of 
    $\ell$.
    See Figure~\ref{fig:larged}(d). 
    Let $n_1 = |\ell^- \cap S|$ and $n_2 = |\ell^+ \cap S|$ where $\ell^+$ and $\ell^-$
    corresponds to the halfplane above and below $\ell$.
    W.l.o.g, assume $n_1 \ge n_2$.
    Consider a point $p$ with $\omega_{q;S}(p) \ge M$.
    We have
	\[M = \frac{C^{-1} \varepsilon^2 \sigma_{S}(q)}{\log n}  = \Omega\left(  \frac{C^{-1} \varepsilon^2 |S|\tau^2_{S}(q)}{\log n}\right)\]
	and thus we must have
	\[\Omega\left(  \frac{C^{-1} \varepsilon^2 |S|\tau^2_{S}(q)}{\log n}\right)  \le \omega_{q;S}(p) \le n_1 n_2 \le |S| n_2   \]
    and thus
    \[ n_2 = \Omega\left( \frac{\varepsilon^2\tau^2_S(q)}{\log n}\right).\]
    Let $r$ be the number of points below $\dual{v'}$.
    On the other hand, we have 
    \[ \omega_{q;S}(p) \ge (n_1 - r)(n_2 - r) \ge n_1n_2 - |S|r.\]
    Since $r =  O(\log^2 n \tau_{S}(q))$, it follows that 
    $n_2 / r = \Omega\left( \frac{\varepsilon^2 \tau_{S}(q)}{\log^3n} \right) = \Omega(\varepsilon^{-1})$
    and thus $n_1n_2$ is  a constant factor approximation of $\omega_{q;S}(p)$.
    This in turn implies $|S| n_2$ is also a constant factor approximation of $\omega_{q;S}(p)$.
    This is the motivation for defining the following concept:
    For a point $p$ below $\dual{v'}$, we call the minimum number of points of ${S}_i$ 
    at either side of line $\overline{pq}$ its {\em dissection value with respect
    to $q$ and ${S}_i$\/} (or its dissection value for short).
    In Figure~\ref{fig:larged}(d), $n_2$ is the dissection value of $p$, if $n_1 \ge n_2$.

    Thus, our goal can be further reduced to the following, that we call \emph{dissection reporting}: 
    For a canonical subgroup $G \subset v^\downarrow$ that contains $g$ points,
    build a data structure that can answer the following:
    given a query point $q$, a threshold $t$, and a line $\dual{v'}$ parallel to $\dual{v}$ such that 
    $G$ lies below $\dual{v'}$, find all points in $G$ with dissection value larger than $t$.
    To find the part of $P_\lar$ in $G$, we simply set $t = \Omega(M/|S|)$ for an appropriate constant
    hiding in the $\Omega(\cdot)$ notation; as we observed, the dissection value times $|S|$ is a
    constant approximation of $\omega_{q;S}(p)$ so this way, we can find a subset that contains
    all the point in $P_\lar \cap G$. 
    However, we will report a few extra points as well; nonetheless, we can observe that the number
    of extra points reported is only a constant factor larger since every point that gets reported
    contributes a lot to the simplicial depth of $q$ and there cannot be too many such points.

    To solve the dissection reporting problem, let $\dual{S}$ be the set of lines dual to $S$. 
    We build a shallow cutting $L'_i$ for the $(\le k_i)$-level of
    $\dual{S}$ (resp. $U'_i$ for the $(\ge k_i)$-level of $\dual{S}$), for $i=c^i, 1
    \le i \le t=O(\log |S|)$ for a large enough constant $c$.
    $L'_i$ is a convex polygonal chain and  if
    $c$ is set large enough, then $L'_i$ lies between the $k_i$-level and the $ck_i$-level of $S$.
    This means that the polygonal chains $L'_i$ will be non-intersecting and $L'_{i-1}$ is contained
    inside $L'_i$.
    Consider a line $\dual{p}$ in dual space (see Figure~\ref{fig:larged}(e)).
    If $\dual{p}$ intersects $L'_i$, then it creates an interval on $L'_i$, marked by the two intersection
    points of $\dual{p}$ with $L'_i$. 
    Let's call this $I_i(\overline{p})$; note that the end points of this interval come from a
    one dimensional domain which can be parameterized in various different ways, such as using
    polar angles from a point inside each polygonal chain.
    Consider $I_i(\overline{p})$ for every point $p \in G$ as well as the
    interval $I_i(\overline{q})$ defined by the
    line dual to the query point. 
    We say two intervals $I_1$ and $I_2$ intersect if they have a point in common and none of the intervals
    fully contains the other.
    Consider an index $i$ such that $I_i(\overline{q})$ intersects the interval $I_i(\overline{p})$  
    but $I_{i-1}(\overline{q})$ does not intersect the interval $I_{i-1}(\overline{p})$.
    This means that the intersection point of the lines $\overline{q}$ and $\overline{p}$ is between
    $L'_{i-1}$ and $L'_i$ (see Figure~\ref{fig:larged}(e)).
	Imagine for the upper chains, the intersection point of the lines $\overline{q}$ and $\overline{p}$ 
	is between $U'_{j-1}$ and $U'_j$ where $j > i$.
    This implies that the dissection value of line $pq$ is $\Theta(k_i)$.

    Motivated by the observation in the above paragraph, we do the following.
    For each pair of indices $(i,j)$, we create a data structure that is capable of finding all the
	lines $\dual{p} \in G$ that the intersection point of $\dual{p}$ and $\dual{q}$
	lies between $L'_{i-1}$ and $L'_i$ and also between $U'_{j-1}$ and $U'_j$.
	If we do this, then we can perform dissection reporting by issuing $O(\log^2 n)$ such queries:
	we first query $(t,t)$, then query $(t,t-1)$ and $(t-1,t)$ and so on.
	Specifically, to find points with dissection value $\Theta(k_i)$ we query
	$(t,i), (t-1,i), \cdots, (i,i)$ as well as $(i,t), (i,t-1), \cdots, (i,i)$.

	It thus remains to show how a query pair $(i,j)$ is answered.
    For every point $p \in G$, we create a tuple of four input intervals corresponding to the
	intersection of $\dual{p}$ with $L'_{i-1}, L'_i, U'_{j-1}$, and $U'_j$.
    The query also defines a tuple of four intervals in a similar fashion.
	The goal is to store the input tuples of intervals in a data structure
    such that given a query tuple of intervals
    we can find all the input interval tuples such that they
	guarantee $\dual{p}$ and $\dual{q}$ intersect between $L'_{i-1}$ and $L'_i$ and also between $U'_{j-1}$ and $U'_j$.
	Let $([a_1,b_1], [a_2,b_2], [a_3,b_3], [a_4,b_4])$ an input tuple of four intervals.
	We create an eight-dimensional  point $(a_1, b_1, \cdots, a_4, b_4)$.

	We now claim the problem can be solved using dominance reporting in eight dimensional space.
	In dominance reporting, a point $p \in \R^d$ is said to \emph{dominate} a point $q \in \R^d$ if
	any coordinates of $p$ is greater than that of $q$. 
	In dominance reporting, we are to store a set of points in a data structure such that given a point $q$,
	we can report all the points dominated by $q$.
	Observe that in our subproblem, we can also map the tuple of intervals corresponding to the query
	to the eight dimensional space. 
	Here, we can express the constraints that 
	$\dual{p}$ and $\dual{q}$ intersect between $L'_{i-1}$ and $L'_i$ and also between $U'_{j-1}$ and $U'_j$
	as inequalities between respective coordinates of the eight dimensional points, meaning,
	the problem can be solved by building a constant number of eight dimensional dominance reporting
	data structures and issuing a constant number of dominance reporting queries.
	
    Dominance reporting queries can be answered in polylogarithmic time,  using  data structure that 
	needs near-linear space and preprocessing time~\cite{AAL3,Chan.Kasper.Mihai.orth,Bentley.80,Chazelle.filtering.search,Chazelle.Guibas.fractional.II}
	polylogarithmic time, we can find the subset of $P_\lar$ in $G$ using polylogarithmic space overhead and polylogarithmic
    query time.

    Thus, in overall, our data structure will use $\tilde{O}(n)$ space and preprocessing time and can 
	answer queries in $\tilde{O}(1)$ time.
\end{proof}

\section{Proof of Lemma~\ref{lem:calc}}\label{sec:calc}
\lemcalc*
\begin{proof}
    We recall the definition of the dissection value from the proof of Lemma~\ref{lem:findlarge}:
    For a point $p \in R$, we call the minimum number of points of $S$ 
    at either side of line $\overline{pq}$ its {\em dissection value with respect
    to $q$ and ${S}$} (or its dissection value for short).
	Let $m$ be the number of points in $h$.

	We build data structure for approximate range counting~\cite{Afshani.Chan.SOCG07}, 
	and halfplane range reporting~\cite{cgl85} on $S$.
	Using these, for every point 
	$p \in R$ we can obtain a constant factor approximation of its dissection value.

	Consider a point $p \in R$.
	Draw the line $\ell=pq$ and 
    let $n_1 = |\ell^- \cap S|$ and $n_2 = |\ell^+ \cap S|$ where $\ell^+$ and $\ell^-$
    corresponds to the halfplane above and below $\ell$
	(very similar to the situation in Lemma~\ref{lem:findlarge}; see also Figure~\ref{fig:larged}(d)).
	Assume $n_1 \le n_2$.
	Let $\Delta_p$ be the set of triangles that contain $q$ and have $p$ as one of their vertices.

	As first case, assume $n_1 \ge  2\varepsilon^{-1}(m +|R|)$.
	At most $|S||R|$ triangles from $\Delta_p$ can have two points from $R$.
	However, $\Delta_p$ contains at least 
	\[ (n_1 - m)(n_2 -m) \ge n_1n_2 - |S|m \] 
	triangles.
	Observe that if $n_1 \ge 2\varepsilon^{-1}(m + |R|)$ then 
	\[ n_1n_2 - |S|m \ge \varepsilon^{-1} |S||R|. \] 
	Thus, the number of triangles in $\Delta_p$ with two points from $R$ is negligible.
	So, we focus on the triangles that have exactly one point from $R$. 
	This is at most $n_1n_2$ and at least $(n_1 - m)(n_2-m)$.
	Again, an easy calculation yield that $n_1n_2 \ge \varepsilon^{-1}m|S|$ and thus
	$n_1n_2$ is a relative $(1+O(\varepsilon))$-approximation of the number of triangles
	in $\Delta_p$.

	Thus, it suffices to handle points with dissection value less than $2\varepsilon^{-1}(m +|R|)$.
	In the rest of this proof, we assume all the points in $R$ have dissection value
	less than this.

\begin{figure}[b]
    \centering
    \includegraphics[scale=0.5]{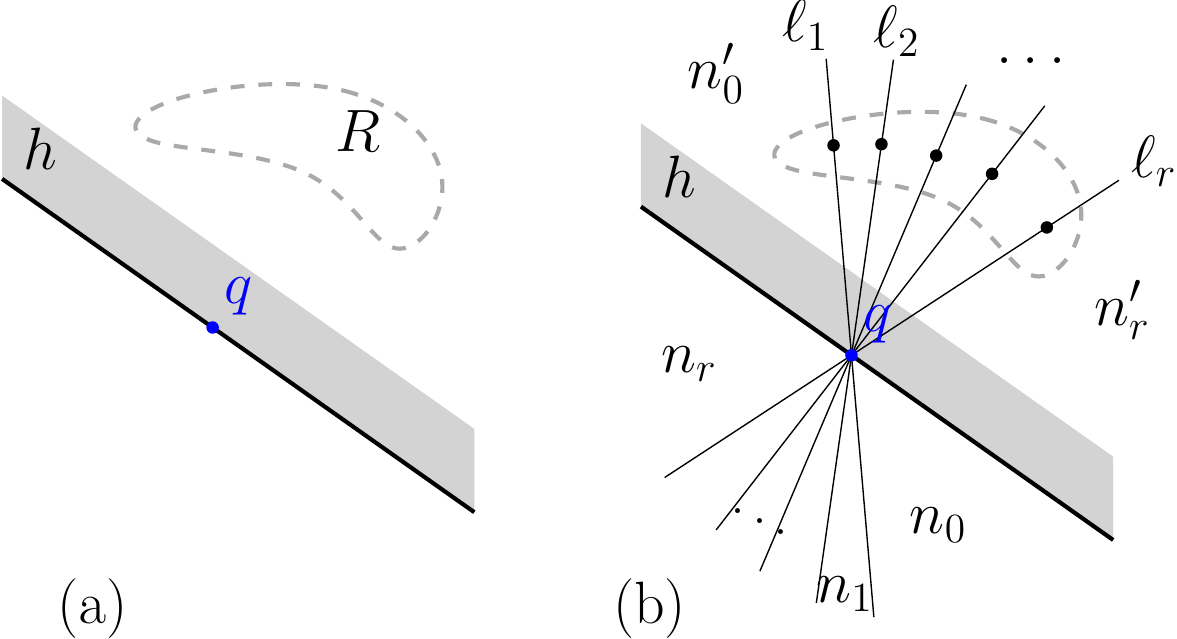}
    \caption{(a) The input configuration for Lemma~\ref{lem:calc}. (b) For each
    point $p \in R$, a line through $p$ and $q$ is created. This subdivides the
    plane into ``wedges'' that are ifinite triangles with $q$ as their only vertex. By knowing
    the exact number of points inside each wedge, we can compute the simplicial depth in $O(r)$ time.}
    \label{fig:small}
\end{figure}

	As a second case,	assume $\tau_S(q) = \tO(\varepsilon^{-1}|R|)$ which also implies $m = \tO(\varepsilon^{-1}|R|)$,
	and that the dissection value of the points in $R$ is at most $2\varepsilon^{-1}(m + |R|) = 
	\tO(\varepsilon^{-1}|R|)$.
	Let $p_1, \cdots, p_r \in R$ be all the points in $R$.
	We create the lines $\ell_1, \cdots, \ell_r$ by connecting $p_i$ to $q$.
	Using halfplane range reporting, we can exactly compute the number of points at one side of 
	each line $\ell_i$ in $\tO(\varepsilon^{-1}|R|)$ time. 
	Over all the points in $R$ this will take $\tO(\varepsilon^{-2}|R|^2)$ time. 
	Given these values, we can compute the number of points that lie inside each ``wedge'' created
	by the lines $\ell_1, \cdots, \ell_r$ (see Figure~\ref{fig:small}(b)) in
	$\tO(\varepsilon^{-2}|R|^2)$ time. 
	Given these values, the total number of triangles that contain $q$ and have one or two 
	points from the set $\left\{ p_1, \cdots, p_r \right\}$ can be counted in $O(r)$ time
    (see~\cite{GilStWi1992,KhullerJo1990}).

	Finally, we assume $\tau_S(q) \ge \varepsilon^{-1}|R| \log^2 n$.
	This case can be handled using similar ideas as the previous two.
	Here,  we claim among triangles in $\Delta_p$, the number of those that
	include two points from $R$ is negligible: the number of such triangles is at most
	$|S| |R|^2$ where at the simplicial depth of $q$ is at least $\Omega(|S|\tau^2_S(q))$
	by Lemma~\ref{lem:SH}.
	We have
	\[|S|\tau^2_S(q) \ge |S| \varepsilon^{-2}|R|^2 = \omega(\varepsilon^{-1}|S| |R|^2)\]
	so we can safely ignore triangles that contain two points from $R$.
	This limits us to triangles that have exactly one point from $R$.
	Let $p_1, \cdots, p_r \in R$ be the points in $R$.
	Again, we create the lines $\ell_1, \cdots, \ell_r$ by connecting $p_i$ to $q$.
	Using halfplane range reporting, we can exactly compute the number of points at each side
	of the line $\ell_i$ in $\tO(\varepsilon^{-1}|R|)$ time. 
	Over all the points in $R$ this will take $\tO(\varepsilon^{-2}|R|^2)$ time. 
	As before, these values enable us to compute the number of triangles containing one point from $R$.

	Combining all these cases, we obtain a relative $(1+O(\varepsilon)$-approximation.
	By scaling $\varepsilon$ by a constant, we can obtain a relative $(1+\varepsilon)$-approximation.
	The total query time is $\tO(\varepsilon^{-2}|R|^2)$.
\end{proof}

\section{Approximation in High Dimensions}
\label{sec:hdim}
In this section, we present two approximation algorithms for simplicial
depth in high dimensions, each with a different worst case scenario. By
combining these strategies, we obtain a constant factor
approximation algorithm with $\hdtime$ running time.

\subsection{Small Simplicial Depth: Enumeration}
Let $P \subset \R^d$ be a set and $q \in \R^d$ a query point. If $\sigma_P(q)$
is small, a simple counting approach that iterates
through all simplices $\Delta \in \Delta_P$ leads to an
efficient algorithm. The key is to construct a graph that contains exactly one
node per simplex $\Delta \in \Delta_P$. Then, counting can be
carried out by a breadth-first search and we avoid looking at subsets of $P$
that do not contain $q$ in their convex hull.
For this, we use the Gale transform to dualize the problem. We
shortly restate important properties of the Gale transform. For more
details see~\cite{thomas06lectures}.
Let in the following $\ori$ denote the origin.

\begin{lemma}\label{lem:gale}
  Let $P = \{p_1,\dots,p_n\} \subset \R^d$ be a point set with $\sigma_P(\ori) >
  0$. Then, there is a set
  $\bar{P} = \{\bar{p}_1,\dots,\bar{p}_n\} \subset \R^{n-d-1}$ such that a
  $(d+1)$-subset $P' \subseteq P$ contains $\ori$ in its convex hull iff $\bar{P}
  \setminus \{ \bar{p}_i \mid p_i \in P'\}$ defines a facet of $\conv(\bar{P})$.
  \qed{}
\end{lemma}

Consider now the graph $G_P(q) = (V, E)$ with $V = \Delta_P$.
Two simplices $\Delta, \Delta'$ are adjacent iff $\Delta'$ can be obtained from
$\Delta$ by swapping one point in $\Delta$ with a different point in $P$. We
call $G_P(q)$ the \emph{simplicial graph} of $P$ with respect to $q$.

\begin{lemma}\label{lem:simgraph}
Let $P \subset \R^d$ be a set of size $n$. Then, $G_P(q)$ is $(n-d-1)$-connected
and $(n-d-1)$-regular.
\end{lemma}

\begin{proof}
We assume w.l.o.g.\ that $q = \ori$.
Let $\Delta, \Delta'$ be two adjacent nodes in $G_P(q)$.
Furthermore let $\bar{P}$ denote the Gale transform of $P$. Set $\bar{\Delta} = \{
\bar{p} \mid p \in P \setminus \Delta \}$ and $\bar{\Delta}' = \{ \bar{p} \mid p \in
P \setminus \Delta'\}$.
By Lemma~\ref{lem:gale}, the two sets $\bar{\Delta}$ and $\bar{\Delta}'$ define
facets of $\conv(\bar{P})$. Since $\Delta$ and $\Delta'$ are adjacent, we have
$|\Delta \cap \Delta'| = d$ and hence $|\bar{\Delta} \cap \bar{\Delta}'| = n - d
- 2$. Thus, the facets defined by $\bar{\Delta}$ and $\bar{\Delta}'$ share a
ridge. Hence, $G_P(q)$ is isomorph to the $1$-skeleton of the
polytope dual to $\conv(\bar{P})$. In particular, this implies that $G_P(q)$ is
$(n-d-1)$-connected. It remains to show that the graph is $(n-d-1)$-regular.
Let $\Delta \in V$ be a node.  Lemma~\ref{lem:change} states that each of the
$n-d-1$ points in $P\setminus \Delta$ can be swapped in, each time resulting in
a distinct simplex.
\end{proof}

Since $G_P(q)$ is connected, we can count the number of vertices using BFS\@.

\begin{lemma}\label{lem:bfs}
  Let $P \subset \R^d$ be a set of size $n$ and $q \in \R^d$ a query point. Then,
  $\sigma_P(q)$ can be computed in $O(n \sigma_P(q))$ time.
\end{lemma}

\subsection{Large Simplicial Depth: Sampling}\label{subsec:sample}
If the simplicial depth is large, the enumeration approach becomes infeasible.
In this case we apply a simple random sampling algorithm.

\begin{lemma}\label{lem:sampling}
  Let $P \subset \R^d$ be a set and $q \in \R^d$ a query point. Furthermore, let
  $\eps, \delta > 0$ be constants and let $m \in \N$ be a parameter. If $\sigma_P(q)
  \geq m$, then $\sigma_P(q)$ can be $(1+\eps)$-approximated in $\tilde{O}(n^{d +
    1} / m)$ time with error probability $O(n^{-\delta})$.
\end{lemma}
\begin{proof}
  Let $\Delta_1,\dots, \Delta_k$ be $k$ random $(d+1)$-subsets of P for 
  $k = \left\lceil\frac{4 \delta n^{d + 1} \log n}{\eps^2 m}\right\rceil$. For each random
  subset $\Delta_i$, let $X_i$ be $1$ iff $q \in \conv(\Delta_i)$ and $0$
  otherwise. We have $\mu = \E[\sum_{i=1}^{k} X_i] =
  k \frac{\sigma_P(q)}{n^{d+1}} = \frac{4 \delta \sigma_P(q) \log n}{\eps^2 m}
  \geq \frac{4 \delta}{\eps^2}\log n$. Applying the Chernoff
  bound, we get $\Pr[ |\sum_{i=1}^{k} X_i - \mu| \geq \eps\mu ] =
  O(n^{-\delta})$. Thus, $\frac{n^{d+1}}{k} X$ is a $(1+\eps)$-approximation of
  $\sigma_P(q)$ with error probability $O(n^{-\delta})$.

  For $d=O(1)$, we can test in $O(1)$ whether a given $(d+1)$-subset of $P$
  contains a point in its convex hull. Hence, the running time is dominated by
  the number of samples.
\end{proof}

\subsection{Combining the Strategies}

\begin{theorem}\label{thm:hdapx}
  Let $P \subset \R^d$ be a set and $q \in \R^d$ a query point. Furthermore,
  let $\eps > 0$ and $\delta>0$ be constants.  Then, $\sigma_P(q)$ can be
  $(1+\eps)$-approximated in $\hdtime$ time with error probability $O(n^{-\delta})$.
\end{theorem}

\begin{proof}[Proof of Theorem~\ref{thm:hdapx}]
  We apply the algorithm from Lemma~\ref{lem:bfs} and stop it once $n^{d/2}$
  nodes of $G_P(q)$ are explored. This requires $O(n^{d/2+1})$ time. If the
  graph is not yet fully explored, we know $\sigma_P(q) \geq n^{d/2}$. We can
  now apply the algorithm from Lemma~\ref{lem:sampling} and compute a
  $(1+\eps)$-approximation in $\tilde{O}(n^{d/2 + 1})$ time with
  error probability $O(n^{-\delta})$.
\end{proof}

\section{Improved Approximation in Three Dimensions}\label{sec:3d}
In this section, we show that the simplicial depth in 3D can be approximated in 
$\tilde{O}(n)$ time, which is clearly optimal up to polylogarithmic factors.
The main ingredients required for our proof are the two-dimensional data structure
from Section~\ref{sec:ldim}, Lemma~\ref{lem:SH}, and Observation~\ref{ob:ray}.

\begin{theorem}
    Let $P$ be a set of $n$ points in 3D. The simplicial depth of a given point 
    $q$ can be approximate in $\tilde{O}(n)$ expected time and with high probability.
\end{theorem}
\begin{proof}
    First, we find a halfspace $h$ that passes through $q$ and contains a subset
    $A \subset P$ of $\Theta(\tau_P(q))$ points on one side of it.
    We can find $h$ using different methods, e.g., by using the general reduction used by Aronov and 
	Har-Peled~\cite{AronovHa2008} which can be summarized as follows: 
    take $2^{-i}$-random samples $S_i$ of $P$, for $i=1, \cdots, \log n$ until we find the
    first index $l$ such that $q$ lies outside the convex hull of $S_l$.
    Repeat this $O(\log n)$ times and let $j$ be the smallest index found during these repetitions.
    Let $h$ be the hyperplane that separates $q$ from $S_j$.
    Testing whether $q$ lies outside the convex hull of $S_j$ and finding $h$ if it does, 
    is a three-dimensional linear programming step and can be done in $O(|S_j|)$ expected time
    so the overall computation time is $\tilde{O}(n)$.
    Let $h^+$ be the side of $h$ that contains $q$.
    A standard application of Chernoff bound shows that with high probability, $h^+$ 
    contains at least $\Omega(\tau_P(q)/\log n)$ points and at most
    $\Omega(\tau_P(q)\log n)$ points of $P$.

    Similar to the technique used in Section~\ref{sec:exact}, consider two parallel
    hyperplanes $h_1$ and $h_2$, but this time parallel to $h$.
    Do a central projection from $q$ and map the points of $P$ onto $h_1$ and $h_2$,
    resulting in point sets $P_1$ and $P_2$ on $h_1$ and $h_2$, respectively.
    However, we can observe that since $h_1$ and $h_2$ are parallel to $h$,
    one of them will contain $\tilde{O}(\tau_P(q))$ points. 
    Let this be $P_1$.

    Any simplex containing $q$ must have at least one point from $P_1$.
	Thus, we can express the simplicial depth of $q$ as the sum of the number of 
	simplices containing exactly one point from $P_1$ (denoted by $\sigma^{(1)}(q)$)
	and the number of simplices containing two or more points from $P_1$ (denoted by $\sigma^{(2+)}(q)$).
    We approximate each term separately.

    To approximate $\sigma^{(1)}(q)$, we build the data structure of
    Theorem~\ref{thm:2d} on $P_2$ in $\tO(n)$ time.  For any point $p \in P_1$, consider the ray
    $\ray{pq}$  and its intersection $p'$ with $h_2$.  The two-dimensional
    simplicial depth $\sigma_{P_2}(p')$ approximates the number of triangles
    that contain $p'$ which is equal to the number of three-dimensional
    simplices that contain $q$ and have only $p$ from $P_1$.  Thus, by issuing
    $|P_1|$ queries to the two-dimensional data structure, we can approximate
    $\sigma^{(1)}(q)$ in $\tilde{O}(|P_1|) = \tO(n)$ time.

    To approximate $\sigma^{(2+)}_P(q)$, notice that 
    $\sigma^{(2+)}_P(q) \le |P_1|^2 n^{2}$ since we are forced to pick at least two points
    from $P_1$.
    On the other hand, by the Lemma~\ref{lem:SH},
	$\sigma_P(q) = \Omega(\tau^3_P(q) n)$.
    Thus, we can use the same approach as in Subsection~\ref{subsec:sample} and
    directly sample simplices.
    However, we sample simplices that have at least two points from $P_2$.
	Using the above inequalities and similar to the analysis used in Lemma~\ref{lem:sampling},
	it suffices to sample 
    $O( \frac{\delta \log n |P_1|^2 n^{2}}{\varepsilon^2\tau^3_P(q) n }) = \tO(n)$ simplices and to obtain
    a relative $(1+\varepsilon)$-approximation with high probability for  $\sigma^{(2+)}_P(q)$.
\end{proof}

\section{An Exact Algorithm in High Dimensions}\label{sec:exact}
In this section we describe a simple strategy to compute the simplicial depth exactly
in $O(n^d \log n)$ time. 
While we do not achieve the conjectured lower bound of $\Omega(n^{d-1})$, we cut down roughly
a factor $n$ compared to the trivial upper bound of $O(n^{d+1})$.
Note that this almost matches the best previous bound of $O(n^4)$ in 4D as well~\cite{ChengOu2001}.

W.l.o.g, assume  $q$ is the origin, $\ori$. 
Our main idea is very simple: consider $d$ points $p_1, \dots, p_d \in P$.
Let $\ray{r_i}$ be the ray that originates from $\ori$ towards $-p_i$.
We would like to count how many points $p \in P$ can create a simplex with $p_1, \dots, p_d$ 
that contains $\ori$.
We observe that this is equivalent to counting the number of points
of $P$ that lie inside the simplex created by rays $\ray{r_1}, \dots, \ray{r_d}$.
We can count this number in polylogarithmic time if we spend $\tilde{O}(n^d)$ time to build
a simplex range counting data structure on $P$.
This would give an algorithm with overall running time of  $\tilde{O}(n^d)$.
We can cut the log factors down to one by employing a slightly more intelligent approach.

We use the following observation made by Gil et al.~\cite{GilStWi1992}.
\begin{ob}\label{ob:ray}
	Let $q$ be a point inside a simplex $a_1\dots a_{d+1}$ and let $a'_i$ be a point
	on the ray $\stackrel \rightarrow{qa_i}$.
	The simplex defined by $a_1\dots a_{i-1}a'_ia_{i+1}\dots a_{d+1}$ contains
	$p$.
\end{ob}

Pick two arbitrary parallel hyperplanes $h_1$ and $h_2$ such that $P$ lies between them.
This can be done easily in $O(n)$ time.
Next, using central projection from $\ori$, we map the points onto the hyperplanes $h_1$ and
$h_2$:
for every point $p_i \in P$, we create the ray $\ray{\ori p_i}$
and let $p'_i$ be the intersection of the ray with $h_1$ or $h_2$.
Thus, the point set $P$ can be mapped to two point sets $P_1$ and $P_2$ where $P_1$ lies on
$h_1$ and $p_2$ lies on $P_2$ and furthermore, by Observation~\ref{ob:ray},
$\sigma_P(q) = \sigma_{P_1 \cup P_2}(q)$.

Now we use the following result from the simplex range counting literature.
\begin{theorem}\cite{Chazelle.et.al.simplex.RS}
    Given a set of $n$ points in $d$-dimensional space, and any constant 
    $\varepsilon > 0$, one can build a data structure of size 
    $O(n^{d+\varepsilon})$ in $O(n^{d+\varepsilon})$ expected preprocessing
    time, such that given any query simplex $\Delta$, the number of points in 
    $\Delta$ can be counted in $O(\log n)$ time.
\end{theorem}

We build the above data structure on $P_1$ and $P_2$. However, since
both of these point sets lie on a $(d-1)$-dimensional flat, the preprocessing time
is $O(n^{d-1+\varepsilon}) = O(n^d)$ if we choose $\varepsilon = 1/2$.
Next, for any $d$ tuples of points $p_1, \dots, p_d$, we create the rays
$\ray{r_1}, \dots, \ray{r_d}$ and the corresponding simplex $\Delta$.
We find the intersection of $\Delta$ in $O(1)$ time with hyperplanes 
$h_1$ and $h_2$ and issue two simplex range counting queries, one in each hyperplane.
Thus, in $O(\log n)$ time, we can count how many simplices contain $\ori$ that are
made by points $p_1, \dots, p_d$. 
We add all these numbers over all $d$ tuples, which counts each simplex containing $\ori$
exactly $(d+1)$ times. 
The number of $d$-tuples is $O(n^d)$ and for each we spend $O(\log n)$ time querying
the data structures.
Thus,  we obtain the following theorem.
\begin{theorem}
    Given a set $P$ of $n$ points in $\R^d$, the simplicial depth of a point $p$ 
    can be computed in $O(n^d \log n)$ expected time. 
\end{theorem}

\section{Complexity}
\label{sec:complexity}
Let $P \subset \R^d$ be a set of $n$ points and $q \in \R^d$ a query point.
If the dimension is constant, then clearly computing $\sigma_P(q)$ can be
carried out in polynomial time. We now consider the case that $d$ is part of the
input. We show that in this case computing the simplicial depth is \#P-complete
by a reduction from counting the number of perfect matchings in bipartite graphs.

\begin{proposition}\label{prop:sph}
  Let $P \subset \R^d$ be a set and $q \in \R^d$ a query point. Then, computing
  $\sigma_P(q)$ is \#P-complete if the dimension is part of the input.
\end{proposition}
\begin{proof}
  Let $G=(V,E)$ be a bipartite graph with $|V|=n$ and $|E|=m$. It is well known
  that computing the number of perfect matchings in $G$ is
  \#P-complete~\cite{Valiant1979}. Let $\mc{P}_H \subset \R^m$ be the perfect matching
  polytope for $G$~\cite[Chapter~30]{GrahamGrLo1995}. It is defined by $m + 2n$
  half-spaces. Furthermore, the number of vertices of
  $\mc{P}_H$ equals the number $k$ of perfect matchings in $G$. Consider now the
  dual polytope $\mc{P}_V \subset \R^m$. It is the convex hull of $m+2n$ points
  $P \subset \R^m$ and the number of facets equals $k$. Let $\bar{P}
  \subset \R^{2n-1}$ be the Gale transform of $P$. By Lemma~\ref{lem:gale},
  there is a bijection between the facets of $\mc{P}_V$ and the
  $(2n-1)$-simplices with vertices in $\bar{P}$ that contain $\ori$ in their
  convex hull. Hence, $\sigma_{\bar{P}}(\ori) = k$.
\end{proof}

Next, we show that computing the simplicial depth is W[1]-hard with respect to
the parameter $d$ by a reduction to \emph{$d$-\Caratheodory}. In
$d$-\Caratheodory, we are given a set $P \subset \R^d$ and have to decide
whether there is a $(d-1)$-simplex with vertices in $P$ that contains $\ori$ in its
convex hull. Knauer et al.~\cite{KnauerTiWe2011} proved that this problem is
W[1]-hard with respect to the parameter $d$.

\begin{proposition}\label{prop:w1h}
Let $P \subset \R^d$ be a set and $q \in \R^d$ a query point. Then, computing
$\sigma_P(q)$ is W[1]-hard with respect to the parameter $d$.
\end{proposition}

\begin{proof}
Assume we have access to an oracle that, given a query point $q$ and a
set $Q \subset \R^d$, returns $\sigma_Q(q)$. We show that
$\#d$-\Caratheodory can be decided with two oracle queries.

Let $k_d$ denote the number of
$(d-1)$-simplices with vertices in $P$ that contain $\ori$ in their convex hulls
and let $k_{d+1}$ denote the number of $d$-simplices with vertices in $P$ that
contain $\ori$ in their interior. Then $\sigma_P(\ori)$ can be written as $(|P| - d)
k_d + k_{d+1}$. We want to decide whether $k_d > 0$.
For each point $p \in P$ let $\tilde{p} \in \R^{d+1}$ denote the
$(d+1)$-dimensional point that is obtained by appending a $1$-coordinate and
similarly, for each subset $P' \subset P$ let
$\tilde{P'}$ denote the set $\{\tilde{p} \mid p \in P'\} \subset \R^{d+1}$. We
denote with
$S$ the set $\{(0,\dots,0,-1)^T, (0,\dots,0,-2)^T\} \subset \R^{d+1}$ and set
$Q = \tilde{P} \cup S$. Again, we want to express $\sigma_Q(\ori)$ as a function of
$k_d$ and $k_{d+1}$.  Let $Q'\subset Q$, $|Q'| = d+2$, be a subset that contains
$\ori$ in its convex hull. Clearly, $Q'$ has to contain a point from $S$. Let
$\tilde{P'} = Q' \cap \tilde{P}$ denote the part from $\tilde{P}$ and let $S' = Q'
\cap S$ denote the part from $S$. By construction of $S$, we have
$(0,\dots,0,1)^T \in \conv(\tilde{P'})$ and hence $\ori \in \conv(P')$. That is,
each $(d+2)$-simplex with vertices in $Q$ that contains $\ori$ in its
convex hull corresponds to either a $d$-simplex or a $(d-1)$-simplex with
vertices in $P$ that contains $\ori$ in its convex hull. Consider now a set $P'
\subset P$ with $|P'| = d+1$ and $\ori \in \conv(P)$. Then, the corresponding set
$\tilde{P'}$ can be extended in two ways to a subset $Q' \subset Q$, $|Q'|=d+2$,
with $\ori \in \conv(Q')$ by taking either point in $S$. On the
other hand, if $P' \subset P$ is a subset of size $d$ with $\ori \in \conv(P')$,
then we can extend $\tilde{P'}$ to a set $Q' \subset Q$, $|Q'| = d+2$, with $\ori \in
\conv(Q')$ by either taking both points in $S$ or by taking one arbitrary point
in $\tilde{P} \setminus \tilde{P'}$ and either point in $S$. Hence, we have
$\sigma_Q(\ori) = 2 k_{d+1} + k_{d-1} + 2 (|P| - d) k_{d-1}$. Since $k_d =
\sigma_Q(\ori) - 2 \sigma_P(\ori)$, we can decide whether $k_d > 0$ with two
oracle queries.
\end{proof}

The following theorem is now immediate. 

\begin{theorem}\label{thm:complexity}
Let $P \subset \R^d$ be a set of $d$-dimensional points and $q \in \R^d$ a query
point. Then, computing $\sigma_P(q)$ is \#P-complete and W[1]-hard with respect
to the parameter $d$.
\end{theorem}

We conclude the section with
constructive result: although computing the simplicial
depth is \#P-complete, it is possible to determine the parity in polynomial-time.

\begin{proposition}
Let $P \subset \R^d$ be a set of points and $q \in \R^d$ a query point.
If $n-d-1$ is odd or $\binom{n}{d}$ is even, then $\sigma_P(q)$ is even.
Otherwise, $\sigma_P(q)$ is odd.
\end{proposition}

\begin{proof}
We assume w.l.o.g.\ that $q$ is the origin.
Since the simplicial graph $G_P(\ori)$ is $(n-d-1)$-regular, the product $(n-d-1)|V| =
(n-d-1)\sigma_P(\ori)$ is even.
If $(n-d-1)$ is odd, $\sigma_P(q)$ has to be even.
Assume now $(n-d-1)$ is even. We construct a new point set $Q$ in
$\R^{d+1}$ similar as in the proof of Proposition~\ref{prop:w1h}.
Let $R$ denote the set $\{(0,\dots,0,-1)^T, (0,\dots,0,2)^T\} \subset \R^{d+1}$
and set $Q = \tilde{P} \cup R \subset \R^{d+1}$, where $\tilde{P}$ is defined as
in the proof of Proposition~\ref{prop:w1h}.
Let us now consider the graph $G_Q(\ori)$. Since $n-d-1$ is even,
$(|Q|-(d+1)-1) = n - d$ is odd. Now, $G_Q(\ori)$ is $(n-d)$-regular and thus
$\sigma_Q(\ori)$ is even.
Let $Q' \subset Q$, $|Q|=d+2$, be a subset that contains the origin in its convex
hull. Then either (i) $R \subset Q'$ or (ii) $Q'$ contains the point $r =
(0,\dots,0,-1)^T \in R$ and $d+1$ points $\tilde{P'} \subseteq \tilde{P}$ with
$(0,\dots,0,1)^T \in \conv(\tilde{P'})$. There are $\binom{n}{d}$ sets $Q'$ with
Property (i) and $\sigma_P(\ori)$ sets $Q'$ with Property (ii). Hence,
we have
$\sigma_Q(\ori) = \sigma_P(\ori) + \binom{n}{d}$ is even
and thus $\sigma_P(\ori)$ is odd iff $\binom{n}{d}$ is odd.
\end{proof}

\paragraph*{Acknowledgements.}
This work was initiated while YS was visiting \\MADALGO in Aarhus. I would like to
thank the working group for their hospitality and for a constructive
atmosphere.

\end{document}